\documentclass[11pt]{amsart}
\usepackage{graphicx}
\usepackage[latin2]{inputenc} 
\usepackage{color}
\usepackage{amssymb}
\usepackage{multirow}

\newtheorem{Th}{Theorem}

\theoremstyle{remark}

\theoremstyle{definition}

\def\mb{\mathbf}

\def\uz{U_\mathcal Z}
\def\fid{|\phi_1\rangle}
\def\p{\mb p}
\def\tr{\textnormal{Tr}}
\def\ESL{\textnormal{ESL}}

\begin{document}

\title{Maximally informative ensembles for SIC-POVMs in dimension 3}
\author{Anna Szymusiak}
\address{Institute of Mathematics, Jagiellonian University, \L ojasiewicza 6, 30-348 Krak\'ow, Poland}
\email{anna.szymusiak@uj.edu.pl}

\begin{abstract}
In order to find out for which initial states of the system the uncertainty of the measurement outcomes will be minimal, one can look for the minimizers of the Shannon entropy of the measurement. In case of group covariant measurements this question becomes closely related to the problem how informative the measurement is in the sense of its informational power. Namely, the orbit under group action of the entropy minimizer corresponds to a maximally informative ensemble of equiprobable elements. We give a characterization of such ensembles for 3-dimensional group covariant (Weyl-Heisenberg) SIC-POVMs in both geometric and algebraic terms. It turns out that a maximally informative ensemble arises from the input state orthogonal to a subspace spanned by three linearly dependent vectors defining a SIC-POVM (geometrically) or from an eigenstate of certain Weyl's matrix (algebraically).

\end{abstract}
\maketitle

\section{Introduction}

The most general description of quantum measurement is given by \emph{positive operator valued measure} (POVM). Mathematically, in the finite case, it is a set of Hermitian positive semidefinite operators $\Pi_j$ summing up to identity $\sum_{j=1}^k\Pi_j=\mathbb I$. The outcomes of the measurement are given by $j\in \{1,\ldots,k\}$ and the probability of obtaining the $j$th outcome, if the system before measurement was in the state $\rho$, is given by $p_j(\rho,\Pi)=\tr(\rho\Pi_j)$. A POVM is called \emph{informationally complete} (IC POVM) if the probabilities $p_j(\rho,\Pi)$ allow us to determine completely the initial state $\rho$. IC POVMs were first introduced in \cite{Pru}. The minimal number of elements of an IC POVM is $d^2$, where $d$ is the dimension of the Hilbert space describing our quantum system. A \emph{symmetric informationally complete} POVM (SIC-POVM) consists of $d^2$ subnormalized rank-one projections $\Pi_j=|\phi_j\rangle\langle\phi_j|/d$ with equal pairwise Hilbert-Schmidt inner products: $\tr(\Pi_i^*\Pi_j)=|\langle\phi_i|\phi_j\rangle|^2/d^2=1/(d^2(d+1))$ for $i\neq j$. The notion of SIC-POVMs has been introduced by Renes et al. \cite{Ren}, but they have been studied previously by Zauner in his PhD Thesis \cite{Zau} under the name of \emph{regular quantum designs with degree $1$}. Since then they have been getting an increased attention in quantum community due to their applications in quantum state tomography \cite{Scott,ZhuEnglert}, quantum cryptography \cite{Renes} or quantum communication \cite{Ore}. The question, whether there exists a SIC-POVM in any finite dimension, is still open. 

We consider the group-covariant SIC-POVMs with respect to the Weyl-Heisenberg (WH) group, see Sect. \ref{WH}. The group-covariance means that there exist a group $G$, its projective unitary representation $G\ni g\to U_g\in\textnormal {U}(d)$ and a surjection $s:G\to \{1,\ldots,k\}$  such that $U_g\Pi_{s(h)}U_g^*=\Pi_{s(gh)}$. We call $|\phi_{s(e)}\rangle$ a \emph{fiducial vector} (where $e$ denotes the neutral element of $G$). Without loss of generality one can assume that $s(e)=1$. The assumption of WH-covariance is not very restrictive since all known SIC-POVMs are group-covariant and most of them are WH-covariant. In particular, if $d$ is prime then group-covariance implies WH-covariance \cite[Lemma 1]{Zhu}.

For which states of the system before the measurement, the uncertainty of the measurement outcomes is minimal? One of the ways to answer this question is to study the Shannon entropy of measurement $\Pi$, defined by:
$$H(\rho,\Pi):=\sum_{j=1}^k\eta(p_j(\rho,\Pi)),$$
for an initial state $\rho$, where  $\eta(x):=-x\ln x$ ($x>0$), $\eta(0)=0$. This quantity  has been already considered, e.g\  in the context of entropic uncertainty
principles \cite{Deu83,KriPar02,Mas07,WehWin10} and any lower bound for the entropy of
measurement can be regarded as an \textsl{entropic uncertainty relation for
single measurement} \cite{KriPar02}.

 The problem of minimizing entropy  is connected with the problem of maximization of the mutual information between
ensembles of initial states (classical-quantum states) and the POVM $\Pi$.
For an ensemble $V:=\{\pi_i,\rho_i\}_{i=1}^l$ of initial states $\rho_i$ with \emph{a priori} probabilities $\pi_i$ the mutual information between $V$ and $\Pi$ is given by
$$I(V,\Pi):=\sum_{i=1}^l\eta\left(\sum_{j=1}^k P_{ij}\right)+\sum_{j=1}^k\eta\left(\sum_{i=1}^l P_{ij}\right)-\sum_{i=1}^l\sum_{j=1}^k\eta(P_{ij}),$$
where $P_{ij}=\pi_i\tr(\rho_i\Pi_j)$ for $i=1,\ldots, l$ and $j=1,\ldots,k$.  The problem of maximization of $I(V,\Pi)$
consists of two dual aspects \cite{Arnetal13,Hol12,Hol13}: maximization over
all possible measurements, providing the ensemble $V$ is given, see, e.g.\
\cite{Hol73,Dav78,Sasetal99,Suzetal07}, and (less explored) maximization over
ensembles, when the POVM $\Pi$ is fixed \cite{Arnetal11,Ore}. We are interested in the second one, which allows us to answer the question how informative the measurement is, by looking for  the quantity called \emph{informational power} \cite{AAS}:
$$W(\Pi):=\sup_V I(V,\Pi).$$
An ensemble that maximizes the mutual information is called \emph{maximally informative} for $\Pi$.  In fact, it is enough to take into consideration ensembles consisting of pure states only \cite{AAS,Ore}. What is more, if $\Pi$ is group-covariant, then the maximizer can be found in the set of group-covariant ensembles, i.e. ensembles of the form $V(\rho):=\{|G|^{-1},U_g\rho U_g^*\}_{g\in G}$, where $\rho$ is a pure state \cite{Ore}. Additionally, the problems of finding the informational power of group-covariant measurement and of minimizing the entropy of such measurement are equivalent since in such situation we have
$$I(V(\rho),\Pi)=\ln |G|-H(\rho,\Pi)=:\tilde H(\rho,\Pi),$$
where $\tilde H(\rho,\Pi)$
is the relative entropy of $\Pi$ with respect to the uniform distribution, i.e.\ the relative entropy (or Kullback-Leibler divergence) of the probability distribution of measurement outcomes with respect to the uniform distribution. Note that $\tilde H$ measures non-uniformity of the distribution of the measurement outcomes and `can be interpreted as a measure of knowledge, as against uncertainty' \cite{SriBan}. Indeed, the greater $\tilde H$ is, the more we know about the measurement outcomes. 

Both minimum entropy of $\Pi$ and its informational power can be interpreted in terms of
 a \emph{quantum-classical channel} $\Phi:\mathcal{S}\left(  \mathbb{C}^{d}\right)  \rightarrow\mathcal{S}\left(
\mathbb{C}^{k}\right)  $ generated by $\Pi$ and  given by $\Phi\left(
\rho\right)  =\sum_{j=1}^{k}\operatorname{tr}\left(  \rho\Pi_{j}\right)
\left|  e_{j}\right\rangle \left\langle e_{j}\right|  $, where $\left(
\left|  e_{j}\right\rangle \right)  _{j=1}^{k}$ is any orthonormal basis in
$\mathbb{C}^{k}$. The \textsl{minimum output entropy} of $\Phi$  is  equal to the minimum entropy of $\Pi$, i.e.\ $\min_\rho S(\Phi(\rho))=\min_\rho H(\rho,\Pi)$ \cite{Sho02}, where $S$ denotes the von Neumann entropy. Moreover, the informational power of $\Pi$ can be identified \cite{Hol12b,Ore} as the \textsl{classical capacity} $\chi(\Phi)$ of the channel $\Phi$, i.e.\ $$W(\Pi)=\chi(\Phi):=\max_{V=\{\pi_i,\rho_i\}}\left\{S\left(\sum_i \pi_i\Phi(\rho_i)\right)-\sum_i \pi_i S(\Phi(\rho_i))\right\}.$$

While working on this paper, we have found out that the informational power for an exemplary SIC-POVM in dimesion 3 had been independently calculated by Dall'Arno et al. \cite{DAr}. Nevertheless, our results work for an arbitrary group covariant SIC-POVM in dimension 3. Moreover, they seem to give a more complete characterization of the maximizers as they provide a deeper insight into their geometric and algebraic structure. In particular, the  group-covariant maximally informative ensembles are described in details.

\section{Weyl-Heisenberg SIC-POVMS}\label{WH}

Let us denote an orthonormal basis in $\mathbb C^d$ by $|e_0\rangle,|e_1\rangle,\dots|e_{d-1}\rangle$. We define the unitary operators $T$ and $S$ as follows:
$$T|e_r\rangle:=\omega^r|e_r\rangle,\quad S|e_r\rangle:=|e_{r\oplus 1}\rangle,$$
where $r=0,\ldots,d-1$, $\oplus$ denotes the addition modulo $d$ and $\omega:=\exp(2\pi i/d)$. We define also $D_\p=D_{(p_1,p_2)}:=\tau^{p_1p_2}S^{p_1}T^{p_2}$, where $\tau:=-\exp(\pi i/d)$ and $\p\in\mathbb Z^2$. We have $$D_\p^*=D_{-\p}$$ and
$$D_{\p}D_{\mb q}=\tau^{\langle\p,\mb q\rangle}D_{\p+\mb q}$$
for all $\p, \mb q\in\mathbb Z^2$, where $\langle\p,\mb q\rangle=p_2q_1-p_1q_2$ is a symplectic form.

The finite Heisenberg group  (called also the finite Weyl-Heisenberg group or the generalized Pauli group) is irreducibly and faithfully represented by the elements of the form $\omega^{p_3}D_{(p_1,p_2)}$, where $p_1,p_2,p_3\in\mathbb Z_d$. The map $\mathbb Z_d\times\mathbb Z_d\ni(p_1,p_2)\mapsto D_{(p_1,p_2)}$ defines the projective unitary representation of $\mathbb Z_d\times\mathbb Z_d$ on $\mathbb C^d$ and matrices $D_\p$ are called Weyl matrices or generalized Pauli matrices.

Let us consider the normalizer of the Weyl-Heisenberg group in the group $\textnormal{UA}(d)$ of all unitary and antiunitary operators on $\mathbb C^d$, the so-called \emph{extended Clifford group} $\textnormal{EC}(d)$. We denote by  $\ESL(2,\mathbb Z_d)$ the group of all $2\times 2$ matrices over $\mathbb Z_{d}$ with determinant $\pm 1$ (mod $d$) and by $\textnormal I(d)$ the group of unitary multiples of identity operator on $\mathbb C^d$. 
Appleby \cite[Thm 2]{App} has shown that for odd dimension $d$ there exists a unique isomorphism $f_E:\ESL(2,\mathbb Z_d)\ltimes(\mathbb Z_d)^2\to \textnormal{EC}(d)/\textnormal{I}(d)$ fulfiling the condition 
$$UD_\p U^*=\omega^{\langle\mb r,\mathcal F\mb p\rangle}D_{\mathcal F\mb p}$$
for any $(\mathcal F,\mb r)\in\ESL(2,\mathbb Z_d)\ltimes(\mathbb Z_d)^2$, $U\in f_E(\mathcal F,\mb r)$ and $\p\in(\mathbb Z_d)^2$.

Let us consider a unitary $U_{(\mathcal F,\mb r)}\in f_E(\mathcal F,\mb r)$, such that $\det\mathcal F=1$,  $\tr\mathcal F\equiv-1$ (mod $d$), $\mathcal F\neq\mathcal I$. One can choose the phase factor of $U_{(\mathcal F,\mb r)}$  in a way that $(U_{(\mathcal F,\mb r)})^3=\mathbb I$ \cite{App}. Such unitary will be called (following Appleby) a {\it canonical order $3$ unitary}.

It has been conjectured by Zauner \cite{Zau} that in every dimension there exists a fiducial vector for some WH-covariant SIC-POVM which is an eigenvector of the canonical order 3 unitary $U_\mathcal Z:=U_{(\mathcal Z,0)}$, where 
$\mathcal Z=\left(\begin{smallmatrix}
0 & -1 \\
1 & -1
\end{smallmatrix}\right)$ has been later referred to as the \emph{Zauner matrix}. This has been strengthen by Appleby \cite{App} to the conjecture that in every dimension there exists a fiducial vector and every such vector is an eigenvector of a canonical order 3 unitary conjugated to $U_\mathcal Z$ (the conjugacy relation is considered up to a phase in the extended Clifford group). Grassl \cite{Gra} gave a counter-example to the latter one in dimension 12, but there are not known counter-examples in the other dimensions. Still, another conjecture by Appleby, that in every dimension there exists a fiducial vector and every such vector is an eigenvector of a canonical order 3 unitary, remains open. Let us recall also that these two conjectures are equivalent in the prime dimensions greater than three, since in these dimensions all canonical order 3 unitaries are in the same conjugacy class \cite{Fla}.  

 We consider a group-covariant SIC-POVM in dimension 3. As mentioned in the introduction, every such SIC-POVM must be group-covariant with respect to the Weyl-Heisenberg group.

\section{Informational power of SIC-POVM}

The first result we present here is strictly connected with the geometry of SIC-POVM in dimension 3 and does not involve any algebraic structure. Though, the assumption of this theorem concerning linear dependency among the vectors defining a SIC-POVM is not at any rate obvious. However, it follows from \cite[Thm~1]{Ben} that this assumption is fulfilled if a SIC-POVM is covariant with respect to the Weyl-Heisenberg group and its fiducial vector is an eigenvector of certain canonical order 3 unitary conjugated to $U_{\mathcal Z}$, which is not a~huge restriction since all known SIC-POVMs in dimension three are of this form. Second theorem gives us a deeper insight into the algebraic structure of some entropy minimizers.

\begin{Th}\label{geom}
Let $\Pi=\{(1/3)|\phi_j\rangle\langle\phi_j|\}_{j=1}^9$ be a SIC-POVM in dimension 3 and let us assume that some three out of nine vectors $|\phi_j\rangle$ are linearly dependent. Then the state  $|\psi\rangle\langle\psi|$, where $|\psi\rangle$ is orthogonal to the two-dimensional subspace spanned by these vectors,  minimizes (resp.\ maximizes) the entropy of $\Pi$ (resp. the relative entropy of $\Pi$). Moreover, all global minimizers (resp.\ maximizers) are of this form.
\end{Th}

\begin{proof}
Let us assume that $|\phi_1\rangle,|\phi_2\rangle$ and $|\phi_3\rangle$ are linearly dependent.
We will consider the Bloch representation of quantum states. We can represent our SIC-POVM on the generalized Bloch set ${\bf B}\subset S^7$ (unit sphere) by the set of vertices of a regular 8-simplex, which we denote by $B:=\{v_1,v_2,\ldots,v_9\}$. Inner products of vectors $\psi_1,\psi_2\in\mathbb C^3$ and the corresponding Bloch vectors $u_1,u_2\in\mathbb R^8$ are related in the following way: $|\langle\psi_1|\psi_2\rangle|^2=(2(u_1\cdot u_2)+1)/3$. In particular, if $|\psi_1\rangle$ and $|\psi_2\rangle$ are orthogonal, then $u_1\cdot u_2=-1/2$. Let us consider the five-dimensional affine subspace $\pi_1$:
\begin{equation*}
\pi_1:=\{u\in\mathbb R^8|u\cdot v_1=u\cdot v_2=u\cdot v_3=-1/2\}
\end{equation*}
and the affine hyperplane $\pi_2$ tangent to the sphere $S^7$ at the point $w:=-(v_1+v_2+v_3)/\|v_1+v_2+v_3\|$:
\begin{equation*}
 \pi_2:=\{w+u_0|u_0\in\mathbb R^8, u_0\cdot w=0\}=\{u\in\mathbb R^8|u\cdot w=1\}.
\end{equation*}
It is easy to check that $w\in\pi_1\subset\pi_2$. Thus $w$ needs necessarily to be the Bloch vector corresponding to $|\psi\rangle$. For $j\in\{4,\ldots,9\}$ we get $w\cdot v_j=1/4$.

We now apply the method based on the Hermite interpolation described in details in \cite{Slo}. Firstly, we redefine the entropy of $\Pi$ to be a function of Bloch vectors: 
$$H_{\bf B}(u):=H(\rho,\Pi)=\sum_{j=1}^{9}\eta\left(\frac{1+2u\cdot v_j}{9}\right)=\sum_{j=1}^{9}h(u\cdot v_j),$$ where $u$ is the Bloch vector corresponding to $\rho$ and $h:[-1/2,1]\to\mathbb R^+$ is given by $h(t)=\eta\left(\frac{1+2t}{9}\right)$. We are looking for the interpolating Hermite polynomial $p$ such that $p(-1/2)=h(-1/2)$, $p(1/4)=h(1/4)$ and $p'(1/4)=h'(1/4)$.  What is crucial here is that $p$ interpolates $h$ \emph{from below}. The degree of $p$ is at most 2, so it is the degree of the polynomial function $P : {\bf B} \to \mathbb R$ given by $P(u)=\sum_{j=1}^9p\left(u\cdot v_j\right)$ for $u\in {\bf B}$.  As $\sum_{j=1}^9v_j=0$, the linear part vanishes. Additionally, since $\sum_{j=1}^9(u\cdot v_j)^2=\frac{9}{8}\|u\|^2$ for any $u\in\mathbb R^8$ (vertices of regular $N$-simplex in $\mathbb R^N$ form a tight frame in $\mathbb R^N$ with bound $N/(N-1)$),  $P$ must be constant on any sphere. Knowing that $P(u)\leq H_{\bf B}(u)$ and $P(w)=H_{\bf B}(w)\ (=\ln 6)$ we conclude that the entropy attains its minimum value (and so the relative entropy $\tilde H$ attains its maximum value) at  $|\psi\rangle\langle\psi|$.

In order to show that all global minimizers of the entropy (and so maximizers of the relative entropy) are of the same form, i.e.\ they are orthogonal to some three out of nine vectors defining SIC-POVM, let us observe that if $\tilde w\in {\bf B}$ is a
global minimizer for $H_{\bf B}$, then $H_{\bf B}(\tilde w)=P(\tilde w)=\ln 6$. In consequence also $h(\tilde w\cdot v_j)=p(\tilde w\cdot v_j)$ for $j=1,\ldots,9$, and so
$\{\tilde w\cdot v_j|j=1,\ldots,9\}\subset\{w\cdot v_j|j=\nolinebreak 1,\ldots,9\}=\{-1/2,1/4\}$, since $p$ agrees with $h$ exactly in the points of interpolation. Under the constraint $\sum_{j=1}^9 \tilde{w}\cdot v_j=0$ we get that $\{\tilde w\cdot v_j|j=1,\ldots,9\}=\{-1/2,1/4\}$ and there are exactly three $j$'s such that $\tilde w\cdot v_j=-1/2$, i.e.\ there are exactly three vectors $|\phi_j\rangle$ orthogonal to $|\tilde \psi\rangle$.\qedhere

\end{proof}

Note that it is possible to give an alternative proof  that the state $|\psi\rangle\langle\psi|$ minimizes the  entropy. It is enough to notice that the lower bound for the Shannon entropy, namely $\ln 6$, provided by Rastegin \cite[Prop. 6]{Ras} is satisfied here, as it was done in \cite[Corollary~2]{DAr}. Moreover, we can consider as well some generalized entropies as it turns out that both the lower bounds for the Tsallis $\alpha$-entropies for $\alpha\in(0,2]$, i.e. $(1-\alpha)^{-1}(6^{1-\alpha}-1)$, given in \cite[Prop. 6]{Ras}  and the lower bound $\ln 6$ for the R\'{e}nyi $\alpha$-entropies for $\alpha\in(0,2]$ given as the corollary from \cite[Prop. 7]{Ras}  are achieved. However, the dimension three can be exceptional, in the sense that the lower bounds may not be satisfied in higher dimensions, as indicated by some preliminary numerical calculations in dimensions four to six. On the other hand, the method based on the Hermite interpolation seems to be applicable also in the higher dimensions.

\begin{Th}
Let $U_{(\mathcal G,\mb q)}$ be a canonical order 3 unitary conjugated (up to a phase in the extended Clifford group) to $\uz$. 
Then the relative entropy of 3-dimensional WH SIC-POVM, whose fiducial vector $\fid$ is an eigenvector of  $U_{(\mathcal G,\mb q)}$ is maximized in the eigenstates of Weyl matrix $D_{\mb s}$, where $\mb s\neq(0,0)$ satisfies $\mathcal G\mb s=\mb s$.
\end{Th}

\begin{proof}
Operators $U_{(\mathcal G,\mb q)}$ and $\uz$ are conjugated if and only if there exists $(\mathcal F,\mb r)\in ESL(2,\mathbb Z_3)\ltimes\mathbb (Z_3)^2$ such that $\mathcal G=\mathcal{FZF}^{-1}$ and $\mb q=(\mathcal I-\mathcal G)\mb r$. 
Now, if $\p$ is a non-zero fixed point of $\mathcal Z$ (thus $\p=(1,2)$ or $\p=(2,1)$), then $\mb s=\mathcal F\p$ is a non-zero fixed point  of $\mathcal G$.
Let us observe that $D_\mb s$ (and so $D_\mb s^2=D_{2\mb s}$) commutes with $U_{(\mathcal G,\mb q)}$:
$$U_{(\mathcal G,\mb q)}D_{\mb s}=\omega^{\langle \mb q,\mathcal G\mb s\rangle}D_{\mb s}U_{(\mathcal G,\mb q)}=D_{\mb s}U_{(\mathcal G,\mb q)},$$
since $\langle \mb q,\mathcal G\mb s\rangle=\langle \mb r-\mathcal G\mb r,\mathcal G\mb s\rangle=\langle\mb r,\mb s\rangle-\langle\mathcal G\mb r,\mathcal G\mb s\rangle=\langle\mb r,\mb s\rangle-\langle\mb r,\mb s\rangle=0$.
We consider the set $S$ consisting of the vectors $\fid$, $D_{\mb s}\fid$ and $D_{2\mb s}\fid$. By commutativity they all belong to the same eigenspace of $U_{(\mathcal G,\mb q)}$. 
Since $\uz$ has two eigenspaces: one-dimensional and two-dimensional, so has $U_{(\mathcal G,\mb q)}$ and we will refer to them by $\mathcal H_1$ and $\mathcal H_2$. Thus vectors from $S$ must be linearly dependent, and since they are not colinear, they span $\mathcal H_2$.
Let us take any $|\psi\rangle\in\mathcal H_1$. It is obviously orthogonal to the above vectors, and so from Theorem \ref{geom} we get that it maximizes the relative entropy of $\Pi$.
There exists a common eigenbasis for $D_\mb s$ and $U_{(\mathcal G,\mb q)}$, thus  $|\psi\rangle$ is also an eigenvector of $D_\mb s$. Since the orbit under the action of WH group of an eigenvector of any operator $D_\p$ is an eigenbasis of this operator \cite[Thm 2.2]{Ban}, the theorem is proven. \qedhere

\end{proof}

It is worth to notice that the above theorems are not equivalent. Let us consider a family of SIC-POVMs parametrized by $t\in[0,\pi/3]$ and generated by the following vectors: $(0,1,-e^{it}\eta^j)$, $(-e^{it}\eta^j,0,1)$, $(1,-e^{it}\eta^j,0)$, $j=0,1,2$, where $\eta:=e^{2\pi i/3}$.
For every $t\in[0,\pi/3]$ the fiducial vector $\fid:=(0,1,-e^{it})$ is an eigenvector of unitary $U_{(\mathcal G,0)}$ for $\mathcal G=\left(\begin{smallmatrix} 1 & 0\\ 1& 1\end{smallmatrix}\right)$. Thus $\fid,D_{\mb s}\fid$ and $D_{\mb s}^2\fid$, where $\mb s:=(0,1)$ is a fixed point of $\mathcal G$, are linearly dependent and the maximum relative entropy is attained in the eigenstates of operator $D_{\mb s}$.  In consequence, a WH-covariant maximally informative ensemble consists of the eigenbasis of $D_{\mb s}$. Nevertheless, there are two special cases: $t=0$ and $t=2\pi/9$, described in details in \cite[Sect.~3]{Ben}. In the former, $\fid$ is an eigenvector of every symplectic canonical order 3 unitary (i.e.~one of the form $U_{(\mathcal F,0)}$)., so the maximum relative entropy is attained at the eigenstates of any Weyl matrix. Thus the maximizers form a set of four MUBs (mutually unbiased bases) \cite[Thm 2.3]{Ban} and there are four WH-covariant maximally informative ensembles, each consisting of different eigenbases of Weyl matrices.  In the latter case, we get additional linear dependencies that do not arise from the eigenspace of any canonical order 3 unitary, e.g. between vectors $\fid,D_{(1,2)}\fid$ and $D_{(2,0)}\fid$, see also \cite{Ben}. It turns out that the orbit under the action of WH group of the vector orthogonal to this additional subspace consists of three MUBs  and together with the maximizers described in Theorem 2 form a set of four MUBs. Therefore there are two WH-covariant maximally informative ensembles: one consisting of eigenbasis of $D_{\mb s}$ and second one consisting of three MUBs indicated above.

\section{Final remarks}

The natural question is whether the analogues of the results presented above may hold in higher dimensions. Dang et al. \cite[Thms 1-3]{Ben} have shown that some linear dependencies  may arise in any dimension $d$, although if $d$ is not divisible by three we may not found a fiducial vector in a proper eigenspace of a canonical order three unitary. Some numerical calculations in dimensions four to six show that one can find a vector of higher relative entropy than the ones predicted by the potential analogues of theorems presented in this paper. 
However, it is still possible that the global maximizers of the relative entropy of WH SIC-POVMs possess some special algebraic properties that may be revealed during further study.

At the end, let us take a look at a SIC-POVM in dimension 2 (vertices of the tetrahedron in the Bloch representation). It is known \cite{AAS,DAr,Ore,Slo} that in this case the maxima of the relative entropy are located in the states constituting the dual SIC-POVM. From the algebraic point of view, all these 8 states lie on the single orbit under the action of the extended Clifford group $\textnormal {EC}(2)$, in consequence, they have the same stabilizers (up to conjugacy). Moreover, these stabilizers are maximal, \cite{Slo}. Note that in dimension 3 the maximizers are not located on the orbit of SIC-POVM, but they have larger stabilizers in the extended Clifford group $\textnormal {EC}(3)$ than the elements of SIC-POVM with the only exception appearing when the fiducial vector is an eigenvector of every symplectic canonical order 3 unitaries (case $t=0$ in the example presented at the end of the previous section).

\section{Acknowledgements}

The author is grateful to Wojciech S\l omczy\'{n}ski for helpful discussions and valuable suggestions for improving this paper. This research was supported by Grant No. N N202 090239 of the Polish Ministry of Science and Higher Education.

\end{document}